\documentclass[conference]{IEEEtran}
\IEEEoverridecommandlockouts

\usepackage{cite}
\usepackage{amsmath,amssymb,amsfonts,amsthm,mathtools,bm}
\usepackage{graphicx}
\usepackage{textcomp}
\usepackage{xcolor}
\usepackage{booktabs}
\usepackage{array}
\usepackage{enumitem}
\usepackage{microtype}
\usepackage{placeins}
\usepackage{url}
\usepackage{hyperref}
\hypersetup{colorlinks=true,linkcolor=blue,citecolor=blue,urlcolor=blue}

\newtheorem{theorem}{Theorem}
\newtheorem{proposition}{Proposition}
\newtheorem{corollary}{Corollary}
\newtheorem{definition}{Definition}
\newtheorem{remark}{Remark}

\newcommand{\X}{\mathcal X}
\newcommand{\Y}{\mathcal Y}
\newcommand{\R}{\mathbb R}
\newcommand{\Z}{\mathbb Z}
\newcommand{\E}{\mathbb E}
\newcommand{\Prob}{\mathbb P}

\newcommand{\pos}[1]{\left[#1\right]_+}
\newcommand{\1}{\mathbf 1}
\newcommand{\calR}{\mathcal R}

\begin{document}

\title{Sparse Kernel Mechanisms for Locally Differentially Private Discrete Channels}

\author{%
\begin{tabular}{ccc}
\begin{minipage}[t]{0.31\textwidth}
\centering
Amirreza Zamani\\
\textit{Department of Information Science and Engineering}\\
\textit{KTH Royal Institute of Technology}\\
Stockholm, Sweden\\
amizam@kth.se
\end{minipage}
&
\begin{minipage}[t]{0.31\textwidth}
\centering
Parastoo Sadeghi\\
\textit{School of Engineering and Technology}\\
\textit{UNSW Canberra}\\
Canberra, Australia\\
p.sadeghi@unsw.edu.au
\end{minipage}
&
\begin{minipage}[t]{0.31\textwidth}
\centering
Mikael Skoglund\\
\textit{Department of Information Science and Engineering}\\
\textit{KTH Royal Institute of Technology}\\
Stockholm, Sweden\\
skoglund@kth.se
\end{minipage}
\end{tabular}}

\maketitle

\begin{abstract}
We study sparse locally private discrete mechanisms generated by a nonnegative kernel and an input-dependent admissible output support. The support set is intentionally small relative to the ambient alphabet and acts as a mechanism-design parameter. This formulation covers metric-ball supports, sparse exponential mechanisms, sparse staircase mechanisms, sparse randomized response, sparse additive kernels such as Skellam perturbations, and Poisson-binomial-inspired bounded-support channels. We give a general exact characterization of pure and approximate local differential privacy for this sparse-kernel class. Pure local differential privacy forces all supports to coincide, so genuinely input-dependent sparse supports are incompatible with pure privacy. In the approximate regime, the privacy defect decomposes exactly into support leakage and overlap excess loss. We instantiate the formula for several discrete mechanism families. For sparse staircase mechanisms, a shell-ratio condition eliminates overlap excess loss. For graph-metric supports, overlap of metric balls is necessary for nontrivial privacy. For sparse randomized response, the support-leakage term has a closed form in terms of support mismatch. For generic radius-truncated additive kernels and, in particular, sparse Skellam mechanisms, we obtain exact finite-support formulas in terms of the Skellam CDF and an exact Bessel-ratio overlap condition. Numerical evaluations illustrate how support radius and kernel shape separately affect support leakage and overlap excess loss.
\end{abstract}


\section{Introduction}

Local differential privacy (LDP) requires each user to randomize its data before any collection takes place, thereby avoiding a trusted curator. Many standard local randomizers assign positive probability to the entire output alphabet. For discrete domains this can be inefficient: if the private symbol is $x$, outputs that are far from $x$ may be implausible, low utility, or expensive to communicate. This motivates a sparse-support viewpoint in which the local channel randomizes $x$ only over a small admissible set $S(x)\subseteq\Y$.

The support restriction serves several roles. First, it reflects the geometry of the problem by concentrating probability on plausible surrogates for the private input. Second, it can reduce distortion and communication by excluding far-away outputs. Third, it changes the privacy accounting itself. If $y\in S(x)\setminus S(x')$, then observing $Y=y$ rules out $x'$. Thus support mismatch is not only a modeling detail; it is a privacy-relevant object. Under pure LDP, such mismatch is impossible unless all supports coincide. Under approximate LDP, mismatch is possible, but it must be paid for through the $\delta$ budget.

A previous paper studied this idea for the sparse discrete-Laplace and sparse Gaussian kernels~\cite{zamani2026sparseLDP}. That work developed explicit radius-truncated privacy--sparsity tradeoffs for those two parametric families. The present paper has a different emphasis. We do not revisit the detailed Laplace/Gaussian radius-truncation analysis and we do not include the separate optimized sparse randomized-response design problem. Instead, we isolate the general sparse-kernel mechanism class and show how the same exact privacy accounting applies to a wider set of discrete kernels: score-based exponential mechanisms, shell-based staircase mechanisms, sparse randomized response as a finite-alphabet piecewise-constant channel, Skellam perturbations, and Poisson-binomial-inspired bounded-support channels.

\subsection{Positioning and novelty}
The main novelty of this paper is to use hockey stick divergence sum. This is useful because it makes sparse mechanism design modular. In fact, the standard hockey-stick divergence sum is partitioned into support leakage and overlap excess loss. The contribution is the resulting sparse-support accounting and its kernel-specific consequences.
Once the support-leakage term and the overlap excess-loss term are separated, different mechanisms can be studied through their own kernel likelihood-ratio structure on the overlap. This point is useful for discrete mechanisms where supports are intentionally restricted by design, for example because the range is finite, communication is limited, or only nearby outputs are semantically meaningful.

The paper develops three concrete consequences of this viewpoint. First, for sparse exponential mechanisms, the usual utility-score law becomes a support-restricted kernel and the privacy condition separates the score difference from the support mismatch. Second, for sparse staircase mechanisms on an ambient metric space, a simple shell-ratio condition eliminates the overlap excess-loss contribution for all nearby input pairs. Third, for radius-truncated sparse Skellam mechanisms, the support-leakage term admits an exact CDF expression and the overlap term is governed by an exact Bessel-ratio condition. These are kernel-specific consequences of the same sparse-support accounting.

The present paper is deliberately complementary to a separate statistical-optimization direction for support-constrained randomized response. Here randomized response is used only as one mechanism instance illustrating the decomposition and the support-mismatch effect. The optimization of the randomized-response probabilities and support radius, together with downstream estimation risk, is left to a separate mechanism-design treatment. This keeps the present paper focused on the information-theoretic sparse-kernel framework and on non-Laplace/non-Gaussian discrete-noise instances.

\subsection{Related work}

Differential privacy was introduced by Dwork, McSherry, Nissim, and Smith~\cite{dwork2006calibrating}; the local model is discussed in the monograph of Dwork and Roth~\cite{dwork2014book}. The local model is closely related to randomized response~\cite{warner1965randomized} and practical locally private systems such as RAPPOR~\cite{rappor2014}. Sparse estimation under local privacy has been studied at the statistical-estimation level~\cite{acharya2021sparse}, and minimax lower and upper bounds under LDP are developed in~\cite{duchi2014local,duchi2018minimax}. These works study statistical estimation performance; our focus is the exact privacy effect of imposing an explicit support map on the channel.
Sparse point-wise privacy leakage and its fundamental limits are studied in~\cite{zamani2026pointwise}.

At the level of mechanism families, metric-based mechanisms appear in geo-indistinguishability and local counting-query privacy~\cite{andres2013geo,kacem2018geometric}. The exponential mechanism of McSherry and Talwar~\cite{mcsherry2007mechanism} provides a score-based sampling template. Staircase mechanisms are optimal or extremal for broad classes of utility objectives in central and local privacy settings~\cite{geng2014optimal,kairouz2014extremal}. Discrete-noise constructions such as Skellam and Poisson-binomial mechanisms arise in distributed and federated privacy~\cite{agarwal2021skellam,chen2022poisson}. Finite-support perturbation laws also arise in statistical disclosure control, including the ABS perturbation analysis of Sadeghi and Chien~\cite{Sadeghi_Chien_2024}. 
\cite{Sadeghi_Chien_2024} considered centralized DP where the relation between datasets is restricted to neighboring datasets for counting queries. Whereas in this paper we generalize this to non-neighboring but still relation restricted inputs.

Information-theoretic views of differential privacy through divergences and hypothesis testing are also relevant. In discrete spaces, approximate DP is equivalent to a hockey-stick, or $E_\gamma$, divergence bound. Asoodeh, Aliakbarpour, and Calmon connect LDP to the contraction of $E_\gamma$-divergence~\cite{asoodeh2021ldpcontraction}. We use this characterization as a starting point and specialize it to channels with input-dependent supports.

\subsection{Contributions}

The contributions are as follows.
\begin{enumerate}[label=(\roman*)]
    \item We give a sparse-kernel LDP characterization for arbitrary nonnegative kernels and input-dependent support maps.
    \item
    We show that approximate privacy loss has two distinct sources: outputs that are possible under one input but not the other, and excessive likelihood ratios on outputs shared by both inputs. This gives an exact privacy-defect formula, including for range-limited privacy constraints.
    \item We instantiate the theory for sparse exponential, staircase, randomized-response, Skellam, and Poisson-binomial-inspired channels.
    \item For sparse staircase mechanisms, we give a metric shell-ratio condition under which the overlap excess-loss contribution vanishes.
    \item For radius-truncated sparse Skellam mechanisms, we derive exact CDF support-leakage formulas and exact Bessel-ratio overlap conditions.
    \item We provide exact finite-sum numerical evaluations that show how support radius and kernel shape separately affect support leakage and overlap excess loss.
\end{enumerate}

\section{Preliminaries}

Let $\X$ and $\Y$ be finite or countably infinite alphabets. Throughout, $\varepsilon\ge0$ and $0\le\delta\le1$; radii and separations on $\Z$ are integers. This includes the integer-line examples below. A channel $M(\cdot\mid x)$ is pure $\varepsilon$-LDP if, for all $x,x'\in\X$ and all $A\subseteq\Y$,
\begin{equation}
M(A\mid x)\le e^\varepsilon M(A\mid x').
\end{equation}
It is $(\varepsilon,\delta)$-LDP if
\begin{equation}
M(A\mid x)\le e^\varepsilon M(A\mid x')+\delta.
\end{equation}
Since $\Y$ is discrete, approximate LDP is equivalent to the hockey-stick, or $E_\gamma$, divergence bound with $\gamma=e^\varepsilon$~\cite{dwork2014book,asoodeh2021ldpcontraction}:
\begin{align}\label{eq:hockey-itw}
&\sup_{A\subseteq\Y}\{M(A\mid x)-e^\varepsilon M(A\mid x')\}
=\nonumber\\&
\sum_{y\in\Y}\pos{M(y\mid x)-e^\varepsilon M(y\mid x')}
\le\delta.
\end{align}

We also use a relation-restricted version of local privacy. For a symmetric privacy relation $\calR\subseteq\X\times\X$, a channel satisfies $(\varepsilon,\delta)$-LDP on $\calR$ if the same condition is required only for $(x,x')\in\calR$. On the integer line we write
\begin{equation}
\calR_H=\{(x,x'):1\le |x-x'|\le H\}
\end{equation}
and call this range-$H$ LDP. The usual all-pairs definition corresponds to $\calR=\X\times\X$.

The following definition fixes the ordered privacy-defect notation used throughout the paper.
\begin{definition}
For a discrete channel $M$, define the ordered $\varepsilon$-privacy defect
\begin{equation}
\Delta^M_\varepsilon(x,x'):=\sum_{y\in\Y}\pos{M(y\mid x)-e^\varepsilon M(y\mid x')}.
\end{equation}
When the channel is clear from the context, we write $\Delta_\varepsilon(x,x')$.
\end{definition}

The following remark clarifies why the numerical and metric examples use a range-restricted privacy relation.
\begin{remark}
In the integer-line examples, $\X=\Y=\Z$. For mechanisms with finite-radius supports on this infinite alphabet, global all-pairs approximate LDP is trivial in the sense that the worst-case defect is one: sufficiently distant supports are disjoint. Therefore numerical examples below are explicitly interpreted on a relation $\calR_H$. This is the same distinction used in metric privacy, local counting-query privacy, and radius-truncated mechanism design.
\end{remark}

\section{Sparse Kernel Mechanisms}

The following definition introduces the sparse-kernel mechanism class studied in the paper.
\begin{definition}\label{def:sparse-kernel-itw}
Let $K_\theta:\X\times\Y\to[0,\infty)$ be a nonnegative kernel indexed by a parameter $\theta$, and let $S(x)\subseteq\Y$ be nonempty for each $x$. Define
\begin{equation}\label{eq:sparse-kernel-itw}
M_\theta(y\mid x)=\frac{K_\theta(x,y)}{N_x(\theta)}\1\{y\in S(x)\},
\end{equation}
where, $N_x(\theta)=\sum_{y'\in S(x)}K_\theta(x,y').$
We assume $K_\theta(x,y)>0$ on $S(x)$ and $N_x(\theta)<\infty$. The set $S(x)$ is the admissible output support for input $x$, and $|S(x)|$ is the sparsity level.
\end{definition}

The next theorem is the main accounting result. It is intentionally stated for a general kernel, because all later mechanisms are obtained by changing $K_\theta$ and $S(x)$.

The following theorem gives the exact pure and approximate LDP characterization for sparse kernel mechanisms.
\begin{theorem}\label{thm:master-itw}
For the sparse kernel mechanism in \eqref{eq:sparse-kernel-itw}, pure $\varepsilon$-LDP holds if and only if $S(x)=S(x')$ for all $x,x'$ and, on the common support $S$,
\begin{equation}\label{eq:pure-itw}
\sup_{x,x'}\sup_{y\in S}\left[
\log\frac{K_\theta(x,y)}{K_\theta(x',y)}+\log\frac{N_{x'}(\theta)}{N_x(\theta)}
\right]\le\varepsilon.
\end{equation}
Moreover, $(\varepsilon,\delta)$-LDP holds if and only if $\sup_{x,x'}\Delta^{M_\theta}_\varepsilon(x,x')\le\delta$, and
\begin{align}\label{eq:decomp-itw}
\Delta^{M_\theta}_\varepsilon(x,x')
&=
\underbrace{\sum_{y\in S(x)\setminus S(x')}
\frac{K_\theta(x,y)}{N_x(\theta)}}_{\text{support leakage}}
\notag\\
&\quad+
\underbrace{\sum_{y\in S(x)\cap S(x')}
\left[
\frac{K_\theta(x,y)}{N_x(\theta)}-e^\varepsilon\frac{K_\theta(x',y)}{N_{x'}(\theta)}
\right]_+}_{\text{overlap excess loss}}.
\end{align}
On a symmetric relation $\calR$, approximate privacy is characterized by the same decomposition with the supremum restricted to $(x,x')\in\calR$. Pure privacy requires $S(x)=S(x')$ for each $(x,x')\in\calR$ and the bound in \eqref{eq:pure-itw} on that pair's common support. Thus supports must coincide within each connected component of the graph induced by $\calR$, but may differ between components.
\end{theorem}

\begin{proof}
Write $P(y)=M_\theta(y\mid x)$ and $Q(y)=M_\theta(y\mid x')$ for a fixed ordered pair. Since the kernel is positive on the nonempty set $S(x)$ and $N_x(\theta)<\infty$, we have $0<N_x(\theta)<\infty$, and $S(x)$ is exactly the support of $P$.

For pure privacy, apply the event inequality to the singleton $\{y\}$. If $y\in S(x)\setminus S(x')$, then $P(y)>0=e^\varepsilon Q(y)$, a contradiction. Reversing the ordered pair proves $S(x)=S(x')$. On this common support,
\[
\log\frac{P(y)}{Q(y)}
=\log\frac{K_\theta(x,y)}{K_\theta(x',y)}
+\log\frac{N_{x'}(\theta)}{N_x(\theta)}.
\]
Hence pure privacy implies \eqref{eq:pure-itw}. Conversely, that bound gives $P(y)\le e^\varepsilon Q(y)$ at every point, including the zeros outside the common support. Summing over any event proves pure privacy.

For approximate privacy, put $b_y=P(y)-e^\varepsilon Q(y)$. The series is absolutely summable because $\sum_y|b_y|\le1+e^\varepsilon$. For every event $A$,
\[
\sum_{y\in A}b_y\le\sum_{y:b_y>0}b_y,
\]
and equality is attained by $A_* =\{y:b_y>0\}$. Thus the supremum over events equals $\sum_y[b_y]_+$, proving the hockey-stick characterization also for countably infinite alphabets.

Partition $\Y$ into $S(x)\setminus S(x')$, $S(x)\cap S(x')$, and $\Y\setminus S(x)$. On the first set $Q=0$, so the summand is $P$. On the second set both probabilities have their normalized kernel expressions. On the third set $P=0$, so the positive part is zero. This gives \eqref{eq:decomp-itw}; taking the supremum over ordered pairs proves both directions of the approximate claim.

The same argument applies separately to every pair in $\calR$. Symmetry supplies both directions on each edge; equality of supports then propagates along paths. No privacy constraint is imposed between different connected components.
\end{proof}

The following corollary records the immediate obstruction to pure privacy for genuinely input-dependent supports.
\begin{corollary}
A sparse kernel mechanism with genuinely input-dependent supports cannot satisfy pure $\varepsilon$-LDP for any finite $\varepsilon$. In the approximate regime, support mismatch is possible, but it contributes directly to the $\delta$ budget through the support-leakage term in \eqref{eq:decomp-itw}.
\end{corollary}

\begin{proof}
For two distinct supports, at least one ordered set difference is nonempty. A point in that difference has positive probability under one input and zero under the other, contradicting pure privacy on its singleton event. For approximate privacy, choose the entire event $S(x)\setminus S(x')$ to obtain the necessary lower bound
\[
\delta\ge M_\theta(S(x)\setminus S(x')\mid x).
\]
Theorem~\ref{thm:master-itw} gives the exact additional overlap contribution. On a relation, the obstruction applies only to required pairs; differing supports in distinct connected components do not contradict pure privacy on that relation.
\end{proof}

The following definition names the pointwise privacy loss on the common support of two inputs.
\begin{definition}
For $y\in S(x)\cap S(x')$, define
\begin{equation}
L_M(x,x';y)=\log\frac{M(y\mid x)}{M(y\mid x')}.
\end{equation}
Then the overlap excess-loss term can be written as
\begin{equation}\label{eq:overlap-loss-form}
\sum_{\substack{y\in S(x)\cap S(x'):\ L_M(x,x';y)>\varepsilon}}
M(y\mid x)\bigl(1-e^{\varepsilon-L_M(x,x';y)}\bigr).
\end{equation}
\end{definition}
Indeed, on the overlap $M(y\mid x')=M(y\mid x)e^{-L_M(x,x';y)}$. Factoring out the positive quantity $M(y\mid x)$ shows that a summand is positive precisely when $L_M(x,x';y)>\varepsilon$, which proves \eqref{eq:overlap-loss-form}.

The following remark interprets the two terms in the privacy-defect decomposition as separate design effects.
\begin{remark}
The support map and the kernel have different privacy roles. The support map determines which outputs create support leakage. The kernel determines the likelihood ratios on the overlap. Enlarging supports often decreases leakage by increasing overlap, but it can also expose more points to the overlap comparison. This separation is the basic privacy--sparsity design principle used throughout the paper.
\end{remark}

\subsection{Metric and graph-ball supports}

The support map in Definition~\ref{def:sparse-kernel-itw} can be arbitrary, but many sparse mechanisms are generated by a metric or graph structure. Let $(\Omega,\rho)$ be a metric space with $\X,\Y\subseteq\Omega$. For a radius $r$, define the metric-ball support
\begin{equation}
B_r(x)=\{y\in\Y:\rho(x,y)\le r\}.
\end{equation}
A graph-metric support is the special case in which $\Omega$ is the vertex set of a graph and $\rho=d_G$ is shortest-path distance. Thus $B_r(x)$ is the set of output vertices within graph distance $r$ from input vertex $x$.

The following proposition identifies the basic overlap condition needed for nontrivial privacy with metric-ball supports.
\begin{proposition}\label{prop:ball-overlap-itw}
Consider any sparse kernel mechanism whose supports are $S(x)=B_r(x)$. If there exists $(x,x')\in\calR$ such that $B_r(x)\cap B_r(x')=\varnothing$, then the ordered defect satisfies
\begin{equation}
\Delta_\varepsilon(x,x')=1.
\end{equation}
Consequently, nontrivial $(\varepsilon,\delta)$-LDP on $\calR$ with $\delta<1$ is impossible unless $B_r(x)\cap B_r(x')\ne\varnothing$ for every $(x,x')\in\calR$.
\end{proposition}

\begin{proof}
Set $A=B_r(x)$. Disjointness gives $M(A\mid x)=1$ and $M(A\mid x')=0$, so the event characterization yields $\Delta_\varepsilon(x,x')\ge1$. On the other hand, $[M(y\mid x)-e^\varepsilon M(y\mid x')]_+\le M(y\mid x)$ at each output, giving $\Delta_\varepsilon(x,x')\le1$ after summation. Thus the defect is exactly one and cannot be bounded by any $\delta<1$.
\end{proof}

This proposition supplies a necessary geometric condition for nontrivial privacy. Overlap alone does not certify a prescribed privacy budget: support leakage and overlap excess loss must still sum to at most $\delta$ for every required pair. Sparsity is governed by the ball cardinalities $|B_r(x)|$, which may grow slowly on a path or grid but exponentially on a high-degree tree. Thus the same privacy range $\calR$ can have different sparsity costs depending on the geometry of the output space.

\section{Score, Shell, and Randomized-Response Kernels}

\subsection{Sparse exponential mechanisms}

The exponential mechanism samples outputs according to a utility or score function. Let $u:\X\times\Y\to\R$ be a utility score and let $\eta>0$. The sparse exponential mechanism is
\begin{equation}\label{eq:exp-itw}
M^{\rm exp}_\eta(y\mid x)=\frac{e^{\eta u(x,y)}}{Z_x^{\rm exp}(\eta)}\1\{y\in S(x)\},
\end{equation}
where $Z_x^{\rm exp}(\eta)=\sum_{y'\in S(x)}e^{\eta u(x,y')}.$
Assume $0<Z_x^{\rm exp}(\eta)<\infty$ for every $x$. Here $\eta$ controls how strongly high-utility outputs are favored.

The following proposition specializes the sparse-kernel characterization to score-based exponential mechanisms.
\begin{proposition}
The mechanism in \eqref{eq:exp-itw} satisfies Theorem~\ref{thm:master-itw} with kernel $K_\eta(x,y)=e^{\eta u(x,y)}$. Under common support $S$, it is pure $\varepsilon$-LDP if and only if
\begin{equation}\label{eq:exp-pure-itw}
\sup_{x,x'}\sup_{y\in S}\left[
\eta(u(x,y)-u(x',y))+\log\frac{Z_{x'}^{\rm exp}(\eta)}{Z_x^{\rm exp}(\eta)}
\right]\le\varepsilon.
\end{equation}
Under this common-support assumption, if $Z_x^{\rm exp}(\eta)$ is independent of $x$ and $\sup_{x,x',y}|u(x,y)-u(x',y)|\le\Delta_u$, then $\eta\Delta_u\le\varepsilon$ is sufficient for pure $\varepsilon$-LDP.
\end{proposition}

\begin{proof}
The assumed finite, positive normalizer makes \eqref{eq:exp-itw} a channel, with a strictly positive kernel on $S(x)$. For any common output,
\[
\frac{M^{\rm exp}_\eta(y\mid x)}{M^{\rm exp}_\eta(y\mid x')}
=e^{\eta(u(x,y)-u(x',y))}
\frac{Z_{x'}^{\rm exp}(\eta)}{Z_x^{\rm exp}(\eta)}.
\]
Taking logarithms and applying Theorem~\ref{thm:master-itw} proves the equivalence under common support. If the normalizers agree, the second factor is one. The sensitivity assumption then bounds the log-ratio by $\eta\Delta_u\le\varepsilon$ for every ordered pair and every output in the common support. Summing the resulting pointwise inequalities proves pure privacy. Equal normalizers alone do not remove leakage when supports differ.
\end{proof}

\subsection{Sparse staircase mechanisms}

Staircase mechanisms are optimal or extremal in several differential privacy settings~\cite{geng2014optimal,kairouz2014extremal}. A sparse staircase mechanism assigns a shell weight according to the distance from the input. To make the metric argument precise, let $(\Omega,d)$ be an ambient integer-valued metric space with $\X,\Y\subseteq\Omega$. Fix $r\in\mathbb N$ and shell weights
\begin{equation}
a_0\ge a_1\ge\cdots\ge a_r>0.
\end{equation}
The support around input $x$ is $S_r(x)=\{y\in\Y:d(x,y)\le r\}$, and the sparse staircase mechanism is
\begin{equation}\label{eq:staircase-itw}
M_a(y\mid x)=\frac{a_{d(x,y)}}{A_x}\1\{y\in S_r(x)\},
\end{equation}
where $A_x=\sum_{y'\in S_r(x)}a_{d(x,y')}.$
Assume each $S_r(x)$ is nonempty and $A_x<\infty$. Extend $a_j$ by zero for $j>r$, so the kernel is defined on the entire ambient space.

The following proposition gives a shell-ratio condition that controls the overlap excess loss for sparse staircase mechanisms.
\begin{proposition}\label{prop:staircase-itw}
The sparse staircase mechanism satisfies Theorem~\ref{thm:master-itw} with kernel $K(x,y)=a_{d(x,y)}$ with the stated zero extension. Under common support $S$, it is pure $\varepsilon$-LDP if and only if
\begin{equation}\label{eq:staircase-pure-itw}
\sup_{x,x'}\sup_{y\in S}\left[
\log\frac{a_{d(x,y)}}{a_{d(x',y)}}+\log\frac{A_{x'}}{A_x}
\right]\le\varepsilon.
\end{equation}
If $A_x= A,\ \forall x,$ and
\begin{equation}\label{eq:shell-ratio-itw}
\max_{\substack{0\le j,k\le r\\ |j-k|\le H}}\log\frac{a_j}{a_k}\le\varepsilon,
\end{equation}
then, for every pair $x,x'$ with ambient distance $d(x,x')\le H$, the overlap excess-loss contribution in \eqref{eq:decomp-itw} is zero.
\end{proposition}

\begin{proof}
The zero-extended shell kernel is positive on $S_r(x)$ and has normalizer $A_x$. On a common support its log-likelihood ratio is exactly the expression in \eqref{eq:staircase-pure-itw}, so Theorem~\ref{thm:master-itw} gives necessity and sufficiency.

For the second claim, fix $d(x,x')\le H$ and $y\in S_r(x)\cap S_r(x')$. Set $j=d(x,y)$ and $k=d(x',y)$; then $j,k\in\{0,\ldots,r\}$. The two triangle inequalities give $j\le d(x,x')+k$ and $k\le d(x,x')+j$, hence $|j-k|\le H$. Therefore
\[
\log\frac{M_a(y\mid x)}{M_a(y\mid x')}
=\log\frac{a_j}{a_k}+\log\frac{A_{x'}}{A_x}
\le\varepsilon,
\]
where the normalizer term vanishes by assumption. Every overlap summand is consequently zero. The shell condition ranges over both $(j,k)$ and $(k,j)$, so it also covers the reverse ordered pair. For an empty overlap the sum is zero by convention, although the support leakage is then one.
\end{proof}

The following remark explains how sparse discrete-Laplace mechanisms arise as a special staircase family.
\begin{remark}
Choosing geometric shell weights $a_j=e^{-\lambda j}$ recovers the sparse discrete-Laplace kernel. The staircase family extends it by allowing the shell profile itself to be designed rather than fixed by a single inverse-temperature parameter.
\end{remark}

\subsection{Sparse randomized response as a kernel instance}

Let $S(x)\subseteq\Y$ satisfy $x\in S(x)$ and $|S(x)|=s$. For $p\ge q\ge0$ with $p+(s-1)q=1$, define
\begin{equation}\label{eq:rr-itw}
M^{\rm RR}(y\mid x)=
\begin{cases}
p, & y=x,\\
q, & y\in S(x)\setminus\{x\},\\
0, & y\notin S(x).
\end{cases}
\end{equation}
This is a piecewise-constant sparse channel. For $q>0$, its actual support is $S(x)$ and it is an instance of Definition~\ref{def:sparse-kernel-itw}. If $q=0$, normalization gives $p=1$ and the actual support reduces to $\{x\}$; the defect formula below still applies directly.

The following proposition gives the exact approximate-privacy defect for sparse randomized response.
\begin{proposition}\label{prop:rr-itw}
For $x\ne x'$, the ordered defect of \eqref{eq:rr-itw} is
\begin{align}\label{eq:rr-defect-itw}
\Delta^{\rm RR}_\varepsilon(x,x')=
q\,|S(x)&\setminus(S(x')\cup\{x\})|
\nonumber \\&+
\begin{cases}
p, & x\notin S(x'),\\
\pos{p-e^\varepsilon q}, & x\in S(x').
\end{cases}
\end{align}
If there are at least two inputs and their admissible supports coincide, pure $\varepsilon$-LDP is equivalent to $p\le e^\varepsilon q$. For $q>0$ this is $p/q\le e^\varepsilon$; for $q=0$ no finite pure-privacy guarantee holds.
\end{proposition}

\begin{proof}
Fix $x\ne x'$. Outputs outside $S(x)$ have zero probability under the first input and contribute nothing. Each $y\in S(x)\setminus(S(x')\cup\{x\})$ contributes $q$. At $y=x$, the first probability is $p$ and the second is zero if $x\notin S(x')$, or $q$ otherwise. These give the two cases in \eqref{eq:rr-defect-itw}.

It remains to consider common outputs other than $x$. If $y=x'$, the difference is $q-e^\varepsilon p\le0$; otherwise it is $(1-e^\varepsilon)q\le0$. Here both inequalities use $p\ge q\ge0$ and $\varepsilon\ge0$. There are no further positive contributions, proving the formula even when $q=0$.

Under common admissible support and with at least two inputs, the defect for each distinct pair is $[p-e^\varepsilon q]_+$. Pure privacy is therefore equivalent to $p\le e^\varepsilon q$. When $q>0$, this is the stated ratio condition. When $q=0$, normalization forces $p=1$, and the event $\{x\}$ has probabilities one and zero under distinct inputs, so finite pure privacy is impossible. With only one input the privacy condition is vacuous.
\end{proof}

The following remark clarifies the role of sparse randomized response in this paper.
\begin{remark}
This proposition is only used as a closed-form illustration of support leakage for a piecewise-constant sparse channel. The problem of optimizing $p,q$, and the support radius for statistical utility is not part of this paper; it is a separate support-constrained randomized-response design problem.
\end{remark}

\section{Sparse Discrete-Noise Kernels}

\subsection{Translation-invariant additive kernels}

Before specializing to the Skellam distribution, it is useful to record the generic radius-truncated additive form. Let $\phi:\Z\to[0,\infty)$ be a symmetric pmf with $\phi(k)=\phi(-k)$. For $\X\subseteq\Z$, $\Y=\Z$, and radius $r\ge0$ with $Z_r^\phi>0$, define
\begin{equation}\label{eq:additive-itw}
M^\phi_r(y\mid x)=\frac{\phi(y-x)}{Z_r^\phi}\1\{|y-x|\le r\},
\
Z_r^\phi=\sum_{j=-r}^{r}\phi(j).
\end{equation}
For $x'=x+h$ with integer $h>0$, the truncation windows are $\{-r,\ldots,r\}$ and $\{h-r,\ldots,h+r\}$ after translation. They equal the actual supports when $\phi$ is positive on $[-r,r]$; with internal zeros, the actual supports are the positive-mass subsets, and the formula below follows directly from \eqref{eq:hockey-itw} without treating the full windows as supports in Theorem~\ref{thm:master-itw}.

The following proposition gives the exact defect formula for any translation-invariant additive kernel with radius truncation.
\begin{proposition}\label{prop:additive-itw}
For the mechanism in \eqref{eq:additive-itw}, let $\delta_h^\phi(\varepsilon,r)=\Delta_\varepsilon(0,h)$, computed using the translated channel. If $0<h\le2r$ then
\begin{align}\label{eq:additive-defect-itw}
\delta_h^\phi(\varepsilon,r)
&=\frac{1}{Z_r^\phi}\sum_{k=-r}^{h-r-1}\phi(k)
\notag\\
&\quad+
\frac{1}{Z_r^\phi}\sum_{k=h-r}^{r}
\pos{\phi(k)-e^\varepsilon\phi(k-h)}.
\end{align}
If $h>2r$, then $\delta_h^\phi(\varepsilon,r)=1$.
\end{proposition}

\begin{proof}
Since $Z_r^\phi>0$, the mechanism is well defined even if some masses inside the truncation window vanish. Translation gives $\Delta_\varepsilon(x,x+h)=\Delta_\varepsilon(0,h)$. For $1\le h\le2r$, partition the first input's window as
\begin{align*}
G_0&=\{-r,\ldots,r\},\quad G_h=\{h-r,\ldots,h+r\},\\
G_0\setminus G_h&=\{-r,\ldots,h-r-1\},\\
G_0\cap G_h&=\{h-r,\ldots,r\}.
\end{align*}
On the first interval the second probability is zero, so the contribution is $\phi(k)/Z_r^\phi$. On the second it is $[\phi(k)-e^\varepsilon\phi(k-h)]_+/Z_r^\phi$. Outside the first window the positive part is zero. Summing gives \eqref{eq:additive-defect-itw} directly from the hockey-stick formula, without requiring strictly positive masses throughout the window.

If $h>2r$, the windows are disjoint and the first input's window has probabilities one and zero, so the defect is one. If $h=0$, it is zero since $\varepsilon\ge0$. Finally, the change of variable $k=h-j$ and symmetry of $\phi$ give
\begin{align*}
\Delta_\varepsilon(h,0)
&=\sum_{j\in\Z}[M^\phi_r(j\mid0)-e^\varepsilon M^\phi_r(j\mid h)]_+\\
&=\Delta_\varepsilon(0,h).
\end{align*}
Thus the result depends only on the absolute separation, which justifies maximizing over positive $h$ for the symmetric range relation.
\end{proof}

This proposition shows that many finite-support discrete perturbation laws share the same boundary/overlap structure; the only difference is the mass function $\phi$. The Skellam case below is useful because its support-leakage term can be written using the Skellam CDF and its overlap condition has an exact Bessel-ratio form.

\subsection{Sparse Skellam mechanisms}

The Skellam distribution is a discrete additive analogue of the Gaussian distribution and has been used in distributed private learning~\cite{agarwal2021skellam}. Let $\psi_\mu(k)=\Prob[Z=k]$ be the pmf of a symmetric Skellam random variable with variance parameter $\mu>0$. Equivalently, $Z=N_1-N_2$ with independent $N_1,N_2\sim{\rm Poisson}(\mu/2)$, and
\begin{equation}\label{eq:skellam-pmf-itw}
\psi_\mu(k)=e^{-\mu}I_{|k|}(\mu),
\end{equation}
where $I_n$ is the modified Bessel function of the first kind of order $n$. For $\X,\Y\subseteq\Z$, the sparse Skellam mechanism is
\begin{equation}\label{eq:skellam-general-itw}
M^{\rm Sk}_\mu(y\mid x)=\frac{\psi_\mu(y-x)}{Z_x^{\rm Sk}(\mu)}\1\{y\in S(x)\},
\end{equation}
where $Z_x^{\rm Sk}(\mu)=\sum_{y'\in S(x)}\psi_\mu(y'-x).$
The following proposition applies the sparse-kernel decomposition to the general sparse Skellam channel.
\begin{proposition}
The sparse Skellam mechanism satisfies Theorem~\ref{thm:master-itw} with $K_\mu(x,y)=\psi_\mu(y-x)$. In particular,
\begin{align}
\Delta^{\rm Sk}_\varepsilon(x,x')
&=
\sum_{y\in S(x)\setminus S(x')}
\frac{\psi_\mu(y-x)}{Z_x^{\rm Sk}(\mu)}
\notag\\
&\quad+
\sum_{y\in S(x)\cap S(x')}
\left[
\frac{\psi_\mu(y-x)}{Z_x^{\rm Sk}(\mu)}
-e^\varepsilon\frac{\psi_\mu(y-x')}{Z_{x'}^{\rm Sk}(\mu)}
\right]_+ .
\end{align}
\end{proposition}

\begin{proof}
For any integer $k$, the event $N_1=\max(k,0)$ and $N_2=\max(-k,0)$ has positive probability when $\mu>0$. Thus $\psi_\mu(k)>0$. On any nonempty admissible support, $0<Z_x^{\rm Sk}(\mu)\le1$, since it is a sum of probabilities from a pmf. All assumptions of Definition~\ref{def:sparse-kernel-itw} hold. Substituting this kernel and its normalizer into \eqref{eq:decomp-itw} proves the displayed formula, including when the supports are infinite.
\end{proof}

\subsection{Radius-truncated sparse Skellam mechanisms}

For $r\in\mathbb N$, define $S_r(x)=\{y\in\Z:|y-x|\le r\}$ and
\begin{equation}\label{eq:skellam-radius-itw}
M^{\rm Sk}_{\mu,r}(y\mid x)=\frac{\psi_\mu(y-x)}{Z_r^{\rm Sk}(\mu)}\1\{|y-x|\le r\},
\end{equation}
where $Z_r^{\rm Sk}(\mu)=\sum_{j=-r}^{r}\psi_\mu(j).$
Let $F_\mu(m)=\Prob[Z\le m]$ denote the CDF of the symmetric Skellam law.

The following proposition expresses the radius-truncated Skellam normalizer and support-leakage term using the Skellam CDF.
\begin{proposition}\label{prop:skellam-support-itw}
For \eqref{eq:skellam-radius-itw},
\begin{equation}
Z_r^{\rm Sk}(\mu)=F_\mu(r)-F_\mu(-r-1)=2F_\mu(r)-1.
\end{equation}
For two centers separated by $h>0$ with $0<h\le 2r$, the ordered support-leakage term comparing $0$ against $h$ is
\begin{equation}\label{eq:skellam-leakage-itw}
\delta^{\rm Sk}_{\rm supp}(h)
=\frac{\sum_{k=-r}^{h-r-1}\psi_\mu(k)}{Z_r^{\rm Sk}(\mu)}
=
\frac{F_\mu(r)-F_\mu(r-h)}{2F_\mu(r)-1}.
\end{equation}
If $h>2r$, the two supports are disjoint and the ordered privacy defect equals one.
\end{proposition}

\begin{proof}
The substitution $j=y-x$ makes the normalizer independent of the center. For any integer $m$, symmetry gives
\[
F_\mu(-m-1)=\Prob[Z\ge m+1]=1-F_\mu(m).
\]
Consequently, the probability of $-r\le Z\le r$ is $F_\mu(r)-F_\mu(-r-1)=2F_\mu(r)-1>0$. The strict positivity follows, for example, from $\psi_\mu(0)>0$.

For $1\le h\le2r$, the outputs possible at center zero but not at center $h$ form $\{-r,\ldots,h-r-1\}$. Their unnormalized mass is
\begin{align*}
\sum_{k=-r}^{h-r-1}\psi_\mu(k)
&=F_\mu(h-r-1)-F_\mu(-r-1)\\
&=[1-F_\mu(r-h)]-[1-F_\mu(r)]\\
&=F_\mu(r)-F_\mu(r-h).
\end{align*}
Dividing by the normalizer proves \eqref{eq:skellam-leakage-itw}. The calculation includes $h=2r$, when the overlap consists of the single point $r$. For $h>2r$ the supports are disjoint, so the defect is one.
\end{proof}

The following proposition combines the Skellam support-leakage formula with the overlap excess-loss term.
\begin{proposition}\label{prop:skellam-itw}
For $0<h\le2r$,
\begin{align}\label{eq:skellam-defect-itw}
\delta^{\rm Sk}_h(\varepsilon,\mu,r)
&=
\frac{F_\mu(r)-F_\mu(r-h)}{2F_\mu(r)-1}
\notag\\
&\quad+
\frac{1}{2F_\mu(r)-1}
\sum_{k=h-r}^{r}\pos{\psi_\mu(k)-e^\varepsilon\psi_\mu(k-h)}.
\end{align}
If $h>2r$, then $\delta^{\rm Sk}_h(\varepsilon,\mu,r)=1$.
\end{proposition}

\begin{proof}
Use Proposition~\ref{prop:additive-itw} with $\phi=\psi_\mu$. On the overlap $\{h-r,\ldots,r\}$, both inputs have the same positive normalizer, and therefore
\begin{align*}
&[M^{\rm Sk}_{\mu,r}(k\mid0)-e^\varepsilon M^{\rm Sk}_{\mu,r}(k\mid h)]_+\\
&\qquad=\frac{[\psi_\mu(k)-e^\varepsilon\psi_\mu(k-h)]_+}{Z_r^{\rm Sk}(\mu)}.
\end{align*}
Proposition~\ref{prop:skellam-support-itw} supplies the complementary support-leakage mass and $Z_r^{\rm Sk}(\mu)=2F_\mu(r)-1$, giving \eqref{eq:skellam-defect-itw}. Disjointness handles $h>2r$. Symmetry, as proved for the additive formula, gives the same defect for the reverse ordered pair; no monotonicity in $h$ is needed when taking a finite maximum over the privacy range.
\end{proof}

The following proposition identifies the exact Bessel-ratio condition for positive overlap excess loss in the sparse Skellam mechanism.
\begin{proposition}\label{prop:bessel-itw}
On $S_r(0)\cap S_r(h)$, the pointwise privacy loss is
\begin{equation}\label{eq:skellam-loss-itw}
L_h^{\rm Sk}(k)=\log\frac{M^{\rm Sk}_{\mu,r}(k\mid0)}{M^{\rm Sk}_{\mu,r}(k\mid h)}=
\log\frac{I_{|k|}(\mu)}{I_{|k-h|}(\mu)}.
\end{equation}
The overlap summand at an index $k$ is positive exactly when
\begin{equation}\label{eq:bessel-threshold-itw}
\frac{I_{|k|}(\mu)}{I_{|k-h|}(\mu)}>e^\varepsilon.
\end{equation}
\end{proposition}

\begin{proof}
For $\mu>0$, every Skellam mass is strictly positive. Thus on the overlap the likelihood ratio is well defined. The common truncation normalizer cancels, as does the common factor $e^{-\mu}$ in \eqref{eq:skellam-pmf-itw}, yielding
\[
\frac{M^{\rm Sk}_{\mu,r}(k\mid0)}{M^{\rm Sk}_{\mu,r}(k\mid h)}
=\frac{\psi_\mu(k)}{\psi_\mu(k-h)}
=\frac{I_{|k|}(\mu)}{I_{|k-h|}(\mu)}.
\]
Taking logarithms proves \eqref{eq:skellam-loss-itw}. Since the denominator is positive, the summand at $k$ is strictly positive if and only if this ratio exceeds $e^\varepsilon$. Equality gives a zero summand. Therefore the total overlap term is positive if and only if at least one overlap index satisfies \eqref{eq:bessel-threshold-itw}; it is zero when the overlap is empty or no such index exists.
\end{proof}

The following remark explains why the Skellam analysis uses the exact Bessel-ratio condition rather than a Gaussian-style threshold.
\begin{remark}
The Bessel-ratio condition \eqref{eq:bessel-threshold-itw} is exact. 
A Gaussian approximation to the Skellam pmf is accurate only in an appropriate central-limit regime, i.e., for indices sufficiently close to the mean relative to the scale of the distribution.
For fixed $\varepsilon>0$ and fixed $h$, a formal Gaussian threshold can place $k$ outside that regime. We therefore evaluate the Skellam overlap excess loss through the exact Bessel-ratio condition rather than through a Gaussian-style closed form.
\end{remark}

\subsection{Poisson-binomial-inspired sparse channels}

The Poisson-binomial mechanism in~\cite{chen2022poisson} is motivated by secure aggregation and communication efficiency in distributed privacy. It should not be interpreted here as an additive single-user LDP perturbation. Instead, we use a Poisson-binomial pmf as a bounded-support kernel template for a sparse local channel. Let $B_x$ be an integer-valued Poisson-binomial random variable with pmf $\beta_x(y)=\Prob[B_x=y]$. Define
\begin{equation}\label{eq:pb-itw}
M^{\rm PB}(y\mid x)=\frac{\beta_x(y)}{Z_x^{\rm PB}}\1\{y\in S(x)\},
\
Z_x^{\rm PB}=\sum_{y'\in S(x)}\beta_x(y').
\end{equation}
Assume $S(x)$ is a nonempty subset of $\{y:\beta_x(y)>0\}$. Then $0<Z_x^{\rm PB}\le1$, and $S(x)$ is the actual support. If an initially chosen admissible set includes zero-mass points, remove those points before applying support or log-ratio statements; if no positive-mass point remains, the channel is undefined.

The following proposition applies the same sparse-kernel accounting to Poisson-binomial-inspired bounded-support channels.
\begin{proposition}
The channel in \eqref{eq:pb-itw} satisfies Theorem~\ref{thm:master-itw} with kernel $K(x,y)=\beta_x(y)$. Its exact approximate privacy defect is
\begin{align}
\Delta^{\rm PB}_\varepsilon(x,x')
&=
\sum_{y\in S(x)\setminus S(x')}
\frac{\beta_x(y)}{Z_x^{\rm PB}}
\notag\\
&\quad+
\sum_{y\in S(x)\cap S(x')}
\left[\frac{\beta_x(y)}{Z_x^{\rm PB}}-e^\varepsilon\frac{\beta_{x'}(y)}{Z_{x'}^{\rm PB}}\right]_+.
\end{align}
If supports are common and
\begin{equation}
\sup_{x,x'}\sup_{y\in S}\left[
\log\frac{\beta_x(y)}{\beta_{x'}(y)}+
\log\frac{Z_{x'}^{\rm PB}}{Z_x^{\rm PB}}
\right]\le\varepsilon,
\end{equation}
then the mechanism is pure $\varepsilon$-LDP.
\end{proposition}

\begin{proof}
The support assumption gives a strictly positive kernel on $S(x)$, and $0<Z_x^{\rm PB}\le1$ ensures normalization. Theorem~\ref{thm:master-itw} therefore applies. On $S(x)\setminus S(x')$ the second probability is zero, giving the first sum. On the intersection, substituting the two normalized pmfs gives the second sum. All outputs outside $S(x)$ give zero positive part.

When the actual supports coincide, both pmfs are positive there and
\[
\log\frac{M^{\rm PB}(y\mid x)}{M^{\rm PB}(y\mid x')}
=\log\frac{\beta_x(y)}{\beta_{x'}(y)}
+\log\frac{Z_{x'}^{\rm PB}}{Z_x^{\rm PB}}.
\]
The displayed bound thus implies the pointwise pure-privacy inequality, and summation proves it for every event. In fact, under these support assumptions the bound is also necessary by Theorem~\ref{thm:master-itw}. Without positivity on the stated support, this logarithmic expression may be undefined.
\end{proof}

\section{Numerical Illustrations}

The experiments are exact finite-sum evaluations of the preceding formulas. They are not Monte Carlo simulations. We use the range-$H$ relation $\calR_H=\{(x,x'):1\le |x-x'|\le H\}$ on the integer line and report the worst ordered defect over this range:
\begin{equation}
\delta_H^M=\max_{1\le h\le H}\Delta^M_\varepsilon(0,h).
\end{equation}
Translation invariance and symmetry identify both ordered directions with the same positive separation, so this is an $(\varepsilon,\delta_H^M)$ guarantee on $\calR_H$, not a global all-pairs guarantee on $\Z$. Any finite-radius mechanism has global defect one if arbitrarily distant pairs are considered in the computation of $\delta$.

For each mechanism we also report a first-distance moment
\begin{equation}
R_1=\E[|Y-X|\mid X=x],
\end{equation}
which is independent of $x$ in the translation-invariant examples below. The goal of the plots is to separate support leakage from overlap excess loss and to show how the support radius affects these terms.

\subsection{Calibrated sparse randomized-response baseline}

The first experiment evaluates the sparse randomized-response instance in Proposition~\ref{prop:rr-itw}. For a support radius $t$, let $S_t(x)=\{x-t,\ldots,x+t\}$, so $s=2t+1$. We use the common-support randomized-response calibration
\begin{equation}
p_t=\frac{e^\varepsilon}{e^\varepsilon+2t},
\qquad
q_t=\frac{1}{e^\varepsilon+2t}.
\end{equation}
This calibration satisfies $p_t/q_t=e^\varepsilon$ when all inputs share a common support. In the sparse line-supported setting, however, supports shift with the input. Pure LDP therefore fails unless the supports coincide, and the relevant quantity is the exact approximate defect in \eqref{eq:rr-defect-itw} on the range $1\le h\le H$.

For completeness, Proposition~\ref{prop:rr-itw} gives, for integer $h\ge1$,
\[
\Delta^{\rm RR}_\varepsilon(0,h)=
\begin{cases}
hq_t, & 1\le h\le t,\\
p_t+(h-1)q_t, & t<h\le2t,\\
1, & h>2t.
\end{cases}
\]
Indeed, the left support difference has $h$ points when $h\le2t$. It excludes the truthful point for $h\le t$ and includes it for $h>t$; the truthful overlap excess vanishes since $p_t=e^\varepsilon q_t$. These expressions are nondecreasing in $h$, including both junctions, so the range maximum is attained at $h=H$. Also,
\[
R_1^{\rm RR}=2q_t\sum_{j=1}^t j=q_t t(t+1).
\]
These formulas reproduce the randomized-response columns of Table~\ref{tab:numeric-values-itw}.

Figure~\ref{fig:itw-experiments}(a) fixes $(\varepsilon,H)=(1,3)$. When $s$ is too small, some privacy-neighboring supports are disjoint and the defect equals one. Once the support is large enough to overlap throughout the range, the defect decreases as $s$ grows. At the same time, the expected reporting distance grows because probability mass is spread over more possible outputs. This baseline illustrates the privacy--distortion tension caused by sparse shifting supports, without optimizing $p$ and $q$.

\subsection{Radius-truncated sparse Skellam decomposition}

The second experiment evaluates the radius-truncated sparse Skellam mechanism. For each support radius $r$, Proposition~\ref{prop:skellam-itw} expresses the defect at separation $h$ as
\begin{equation}
\delta_h^{\rm Sk}(\varepsilon,\mu,r)
=\delta_{{\rm supp},h}^{\rm Sk}+\delta_{{\rm ov},h}^{\rm Sk},
\end{equation}
where the support-leakage part is the CDF term in \eqref{eq:skellam-leakage-itw} and the overlap part is the positive-part sum in \eqref{eq:skellam-defect-itw}. We evaluate both terms exactly and report the decomposition at the separation that attains the range-$H$ maximum.

Figure~\ref{fig:itw-experiments}(b) fixes $(\varepsilon,\mu,H)=(1,2,3)$ and varies $s=2r+1$. The support-leakage term decreases rapidly with $s$, because enlarging the support makes the shifted Skellam supports overlap more. The overlap excess-loss term behaves differently: as the support grows, more common points are included in the likelihood-ratio comparison, and some of these points can contribute positive excess loss. In this example the overlap term increases and then stabilizes, while the total defect approaches a plateau. Reporting the two terms separately reveals behavior that would be hidden by the total defect alone.

\begin{figure}[t]
\centering
\includegraphics[width=\columnwidth]{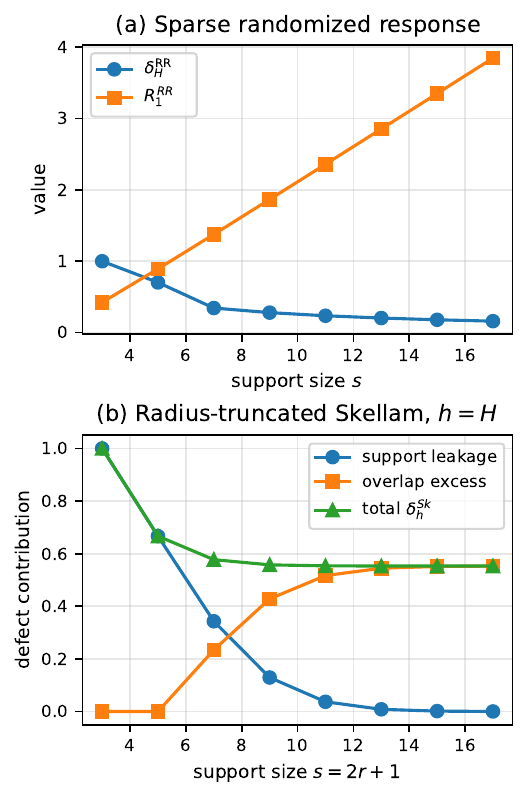}
\caption{Exact range-$H$ numerical evaluations for two sparse-kernel mechanisms. Panel (a) shows a calibrated sparse randomized-response baseline. Panel (b) decomposes the radius-truncated sparse Skellam defect into support leakage and overlap excess loss.}
\label{fig:itw-experiments}
\end{figure}

The numerical procedure is deliberately transparent. For each support size, we first identify the support difference and overlap for each separation $h=1,
\dots,H$. We then evaluate the two terms in the exact decomposition directly. Finally, we take the maximum over $h$. Table~\ref{tab:numeric-values-itw} reports the values used to generate Figure~\ref{fig:itw-experiments}. The randomized-response columns show that the calibrated baseline has a monotone privacy--distortion tradeoff: the defect decreases while $R_1^{\rm RR}$ increases. The Skellam columns show a different phenomenon: support leakage decreases, but the overlap excess-loss contribution increases before stabilizing. Thus total privacy improvement can plateau even when support leakage continues to shrink.

\begin{table}[t]
\centering
\scriptsize
\caption{Exact values used in Figure~\ref{fig:itw-experiments}. The randomized-response baseline uses $(\varepsilon,H)=(1,3)$ and the Skellam experiment uses $(\varepsilon,\mu,H)=(1,2,3)$.}
\label{tab:numeric-values-itw}
\resizebox{\columnwidth}{!}{%
\begin{tabular}{c cc ccc}
\toprule
$s$ & $\delta_H^{\rm RR}$ & $R_1^{\rm RR}$ & $\delta_{\rm supp}^{\rm Sk}$ & $\delta_{\rm ov}^{\rm Sk}$ & $\delta_H^{\rm Sk}$ \\
\midrule
3  & 1.0000 & 0.4239 & 1.0000 & 0.0000 & 1.0000 \\
5  & 0.7023 & 0.8931 & 0.6667 & 0.0000 & 0.6667 \\
7  & 0.3441 & 1.3764 & 0.3431 & 0.2342 & 0.5773 \\
9  & 0.2799 & 1.8660 & 0.1293 & 0.4282 & 0.5575 \\
11 & 0.2359 & 2.3588 & 0.0370 & 0.5167 & 0.5537 \\
13 & 0.2038 & 2.8536 & 0.0084 & 0.5447 & 0.5531 \\
15 & 0.1794 & 3.3496 & 0.0016 & 0.5515 & 0.5530 \\
17 & 0.1603 & 3.8465 & 0.0003 & 0.5528 & 0.5530 \\
\bottomrule
\end{tabular}%
}
\end{table}

The table also illustrates why a sparse-kernel analysis should keep the two defect components separate. If only the total Skellam defect were reported, the plateau near $0.553$ might suggest that the support radius no longer matters. The component decomposition shows a more precise explanation: increasing $s$ almost eliminates support leakage, but the overlap comparison becomes the dominant privacy cost. This behavior is kernel-specific and is not visible from support cardinality alone.

\section{Further Design Implications}

The sparse-kernel formulas give a small set of practical checks. First, the support map should be inspected before the kernel is tuned: if neighboring supports are disjoint, the range-restricted defect is already one. Second, once overlap is guaranteed, the kernel should be chosen so that the likelihood ratios on the overlap are controlled. For shell kernels this leads to shell-ratio constraints; for additive kernels it leads to a ratio condition on the noise pmf; for Skellam kernels it leads to the exact Bessel-ratio condition. Third, the smallest support is not always optimal once overlap loss is included: very small supports cause leakage, while very large supports can expose additional overlap points. The best design is therefore a balance between support size and kernel shape.

These implications are intentionally qualitative in this paper. A separate statistical mechanism-design treatment can take one mechanism family, choose a utility criterion, and optimize all design parameters. The role of the present paper is to provide the exact privacy accounting and kernel-specific formulas needed for such optimizations.

\section{Summary of Mechanism Instances}

Table~\ref{tab:instances-itw} summarizes the sparse-kernel instantiations and the role of each kernel or support structure. The table emphasizes that all mechanisms share the same support obstruction and the same approximate-defect decomposition; only the kernel-specific overlap term changes.

\begin{table}[]
\centering
\small
\caption{Sparse-kernel instances covered in this paper.}
\label{tab:instances-itw}
\begin{tabular}{p{0.24\columnwidth}p{0.31\columnwidth}p{0.34\columnwidth}}
\toprule
Family & Kernel or law & Main consequence \\
\midrule
Metric or graph balls & $S(x)=B_r(x)$ & Feasibility requires overlap of neighboring balls; sparsity is governed by $|B_r(x)|$. \\
Exponential & $e^{\eta u(x,y)}$ & Utility score controls overlap ratios; support mismatch contributes to $\delta$. \\
Staircase & $a_{d(x,y)}$ & Shell-ratio bounds can eliminate overlap excess loss. \\
Randomized response & piecewise $p,q,0$ & Support mismatch has a closed-form defect. \\
Skellam & $\psi_\mu(y-x)$ & Radius truncation gives CDF support leakage and exact Bessel-ratio overlap. \\
Poisson-binomial-inspired & $\beta_x(y)$ & Bounded-support kernels inherit exact support-leakage/overlap decomposition. \\
\bottomrule
\end{tabular}
\end{table}

\section{Discussion and Conclusion}

The sparse-kernel viewpoint separates a local mechanism into two components: the kernel $K_\theta(x,y)$, which controls relative probabilities on the overlap of supports, and the support map $S(x)$, which controls sparsity, plausibility, communication, and support leakage. This separation is exact at the privacy level. Pure LDP does not allow genuinely input-dependent supports. Approximate LDP does allow them, but the price is visible in the support-leakage term of the hockey-stick defect.

The mechanism instances show that this separation is useful beyond the Laplace and Gaussian kernels. Metric-ball supports reveal the role of support overlap and ball growth. Sparse exponential mechanisms use a utility score; sparse staircase mechanisms use shell weights; sparse randomized response uses a piecewise-constant finite-alphabet law; sparse Skellam mechanisms use a discrete additive-noise pmf; and Poisson-binomial-inspired channels use bounded-support kernels from distributed discrete-noise constructions. In each case, the same decomposition determines the exact approximate privacy parameter.

The Skellam analysis also illustrates why kernel-specific work remains necessary after the general decomposition is known. The support-leakage term has a clean CDF form, but the overlap term is governed by a Bessel-ratio condition rather than by the elementary Gaussian threshold. Thus the sparse-kernel framework provides a common language, while each kernel contributes its own overlap structure.

Overall, the results suggest the following design principle. Choose supports that are small enough for utility and communication, but sufficiently overlapping to keep support leakage within the desired $\delta$ budget; then tune the kernel so that the likelihood ratios on the overlap do not create excessive overlap loss. Extending exact optimization results to score-based kernels, shell-weight design, and discrete-noise mechanisms is a natural next step.
\bibliographystyle{IEEEtran}
\bibliography{ITW2027}

\end{document}